\documentclass[
final
]{dmtcs-episciences}

\usepackage[utf8]{inputenc}
\usepackage{subfigure}

\usepackage{amssymb,amsmath,amsthm}
\usepackage{color,enumerate}
\usepackage[ruled,vlined,linesnumbered]{algorithm2e}
\usepackage{relsize}

\newtheorem{theorem}{Theorem}[section]
\newtheorem{lemma}[theorem]{Lemma}
\newtheorem{corollary}[theorem]{Corollary}
\newtheorem{proposition}[theorem]{Proposition}

\newtheorem{problem}[theorem]{Problem}
\newtheorem{remark}[theorem]{Remark}
\newtheorem{example}[theorem]{Example}
\theoremstyle{definition}
\newtheorem{definition}[theorem]{Definition}

\usepackage[round,sort&compress]{natbib}

\author{Didem G\"oz\"upek\affiliationmark{1}
  \and Ademir Hujdurovi\' c\affiliationmark{2,3}\thanks{This work is supported in part by the Slovenian Research Agency (research program P1-0285 and research projects N1-0032, N1-0038, N1-0062, and J1-7051).}
  \and Martin Milani\v c\affiliationmark{2,3}\thanks{This work is supported in part by the Slovenian Research Agency (I0-0035, research program P1-0285 and research projects N1-0032, J1-5433, J1-6720, J1-6743, and J1-7051).}}
\title[Minimal dominating sets and well-dominated lexicographic product graphs]{Characterizations of minimal dominating sets and the well-dominated property in lexicographic product graphs\thanks{The support of a bilateral research project between Slovenia and Turkey, financed by the Slovenian Research Agency (BI-TR/$14$--$16$--$005$) and T\"UB\.{I}TAK (grant no:213M620) is gratefully acknowledged.}}
\affiliation{
  Department of Computer Engineering, Gebze Technical University, Gebze, Kocaeli, Turkey\\
  University of Primorska, UP IAM, Muzejski trg 2, Koper, Slovenia\\
  University of Primorska, UP FAMNIT, Glagolja\v ska 8, Koper, Slovenia}
\received{2017-3-21}
\revised{2017-8-1}
\accepted{2017-8-2}
\begin{document}
\publicationdetails{19}{2017}{1}{25}{3209}
\maketitle
\begin{abstract}
  A graph is said to be well-dominated if all its minimal dominating sets are of the same size. The class of well-dominated graphs forms a subclass of the well studied class of well-covered graphs. While the recognition problem for the class of well-covered graphs is known to be co-NP-complete, the recognition complexity of well-dominated graphs is open.

In this paper we introduce the notion of an irreducible dominating set, a variant of dominating set generalizing both minimal dominating sets and minimal total dominating sets. Based on this notion, we characterize the family of minimal dominating sets in a lexicographic \hbox{product} of two graphs and derive a characterization of the well-dominated lexicographic product graphs. As a side result motivated by this study, we give a polynomially testable characterization of well-dominated graphs with domination number two, and show, more generally, that well-dominated graphs can be recognized in polynomial time in any class of graphs with bounded domination number. Our results include a characterization of dominating sets in lexicographic product graphs, which \hbox{generalizes} the expression for the domination number of such graphs following from works of Zhang et al.~(2011) and of \v{S}umenjak et al.~(2012).
\end{abstract}

\section{Introduction}

Variants of domination play an important role in graph theory. They give rise to theoretically interesting (and often difficult) problems and are widely applicable to model various real-life scenarios, see, e.g.~\cite{MR1605685,MR1605684}. The main subject of the present work is the study of domination in lexicographic product graphs. The works in the literature about domination in various product graphs have been mostly centered around the Cartesian product, where the focus has largely been on Vizing's conjecture (\cite{brevsar2012vizing}). For the lexicographic product graphs, various types of domination were investigated in the literature, including domination (\cite{MR1429806,vsumenjak2012roman}), total domination (\cite{MR2828067}), rainbow domination (\cite{vsumenjak2013rainbow}), Roman domination (\cite{vsumenjak2012roman}), and restrained domination (\cite{MR2828067}). In particular, the works
of~\cite{vsumenjak2012roman} and of~\cite{MR2828067} imply that the value of the domination number of a nontrivial lexicographic product of two graphs can be exactly determined as a function of the domination and total domination numbers of its factors.
In this paper, we expand on these results by studying the \hbox{(inclusion-)minimal} dominating sets and related notions in lexicographic product graphs.

One of the central notions for our study is that of {\it well-dominated graphs}. These are defined as graphs in which all minimal dominating sets have the same size. Well-dominated graphs form a subclass of the more widely studied class of well-covered graphs, defined as graphs in which all maximal independent sets are of the same size~(\cite{MR1254158}, \cite{MR1677797}). Well-dominated graphs were introduced by~\cite{finbow1988well}, who characterized well-dominated graphs of girth at least~$5$ as well as well-dominated bipartite graphs, and showed that within graphs of girth at least~$6$, well-dominated graphs coincide with the well-covered ones.
Not much work has been done on the subject since then. In particular, while the recognition problem for the class well-covered graphs was shown to be co-NP-complete (\cite{MR1161178},~\cite{MR1217991}), the recognition complexity of well-dominated graphs is not known. Characterizations of well-dominated graphs were obtained within the families of block graphs and unicyclic graphs by~\cite{topp1990well}, \hbox{$4$-connected} $4$-regular claw-free graphs by~\cite{gionet2011revision}, planar triangulations by~\cite{finbow2015triangulations}, and graphs without cycles of lengths $4$ and $5$ by~\cite{MR3648208}.

We introduce the notion of an irreducible dominating set, a variant of the notion of a dominating set that forms a common generalization of both minimal dominating sets and minimal total dominating sets (see Section~\ref{sec:irreducible}).
Irreducible dominating sets are important for the characterization of minimal dominating sets in a nontrivial lexicographic product of two graphs, which we develop in Section~\ref{sec:minimal} (Theorem~\ref{thm:MinDom}). Building on this characterization, we derive our main result: a characterization of the well-dominated nontrivial lexicographic product graphs (Theorem~\ref{thm:Well-dominated Lex}).
This characterization motivates the study of well-dominated graphs with domination number two, for which we develop a polynomially testable characterization in Section~\ref{sec:small-gamma} (Theorem~\ref{thm:characterization gamma=2}). More generally, using a connection with the well-known Hypergraph Transversal problem we show that well-dominated graphs can be recognized in polynomial time in any class of graphs with domination number bounded by a constant (Theorem~\ref{thm:const-k}). We conclude the paper with some open questions.

\section{Preliminaries}

All graphs in this paper will be finite, simple, and undirected. An {\em independent set} in a graph is a set of pairwise non-adjacent vertices.
An independent set is said to be {\em maximal} if it is not contained in any larger independent set.
The maximum size of an independent set in a graph $G$ is called the {\em independence number} of $G$ and denoted by $\alpha(G)$.
For $u\in  V(G)$, we denote by $N_G(u)$ the set of all neighbors of $u$ in $G$, and by $N_G[u]$ the closed neighborhood of $u$, that is,  $N_G[u]=\{u\}\cup N_G(u)$. For a set $S\subseteq V(G)$, we write $N_G(S)$ for the set $\cup_{v\in S}N_G(v)$, and
$N_G[S]$ for the set $\cup_{v\in S}N_G[v]$. Note that for every set $S$, we always have $S\subseteq N_G[S]$,
while $S$ is not necessarily a subset of $N_G(S)$.
A vertex $v$ in a graph $G$ is {\em isolated} if $N_G(v) = \emptyset$ and {\em universal} if $N_G[v] = V(G)$.
In all these notations, we may omit the index whenever the graph is clear from the context.
We denote the complement of a graph $G$ by $\overline{G}$ and a cycle on $n$ vertices by $C_{n}$. We denote by $2K_1$ the edgeless graph with exactly two vertices.

A set $D\subseteq V(G)$ is said to be a {\em dominating set} in $G$ if $N[D]=V(G)$; equivalently, if $D\cap N[v]\neq\emptyset$ for all $v\in V(G)$.
The minimum size of a dominating set of a graph $G$ is called the {\em domination number} of $G$ and denoted by $\gamma(G)$.
A {\em minimum} dominating set in $G$ is a dominating set of size $\gamma(G)$.
A {\em total dominating set} in $G$ is a set $D\subseteq V(G)$ such that $N(D) = V(G)$, that is, if every vertex of $G$ has a neighbor in $D$.
(Note that total dominating sets exist only in graphs without isolated vertices.)
The {\em total domination number} of $G$, denoted by $\gamma_t(G)$, is the minimum size of a total dominating set.
A dominating set (resp., a total dominating set) $D$ in $G$ is said to be {\em minimal} if it is minimal with respect to inclusion, that is, $D$ is a dominating set (resp., a total dominating set) that does not contain any smaller dominating set (resp., total dominating set) in $G$. The maximum size of a minimal dominating set of a graph $G$ is called the {\em upper domination number} of $G$ and denoted by $\Gamma(G)$.
A graph $G$ is said to be {\em well-covered} if all its maximal independent sets are of the same size, and {\em well-dominated} if all of its minimal dominating sets are of the same size, that is, if $\gamma(G)=\Gamma(G)$.

Let $S$ be a set of vertices in a graph $G$. For $x\in V(G)$, we say that $x$ is {\em dominated} by $S$ (or that $S$ {\em dominates} $x$) if $N[x]\cap S\neq \emptyset$. Moreover, we say that $x$ is {\em totally dominated} by $S$ (or that $S$ {\em totally dominates} $x$) if $N(x)\cap S\neq \emptyset$, and we say that $x$ is {\em barely dominated} by $S$ if
$x$ is dominated by $S$ and $x$ is not totally dominated by $S$ (or, equivalently, if
$N[x]\cap S = \{x\}$, that is, if $x$ is an isolated vertex in the subgraph of $G$ induced by $S$). For two sets of vertices $S$ and $S'$ in $G$, we say that $S$ {\em totally dominates} $S'$ if every vertex in $S'$ is totally dominated by $S$. In particular, a set $S\subseteq V(G)$ is a total dominating set in $G$ if and only if $S$ totally dominates $V(G)$. For graph theoretic terms not defined here, see, e.g.,~\cite{MR1367739}.

\subsection{Domination number of lexicographic product graphs}

The {\em lexicographic product} of two graphs $G$ and $H$ is the graph $G[H]$ (sometimes denoted also by $G\circ H$) with vertex set $V(G)\times V(H)$, where two vertices $(x_1,y_1)$ and $(x_2,y_2)$ are adjacent if and only if either $x_1x_2\in E(G)$ or $x_1 = x_2$ and $y_1y_2\in E(H)$.
The lexicographic product of two graphs is said to be \emph{nontrivial} if both factors have at least two vertices.
The {\em projection} of a given subset $D\subseteq V(G)\times V(H)$ to the graph $G$ is defined by
$p_G(D)=\{x\in V(G):  (x,y)\in D$ for some $y\in V(H)\}$. For a vertex $x\in p_G(D)$, we define $p_{H,x}(D)=\{y\in V(H): (x,y)\in D\}$. Note that every set $D\subseteq V(G[H])$ can be expressed as the disjoint union
\begin{equation}\label{eq:decomposition-D}
  D = \bigcup_{x\in p_G(D)}(\{x\}\times p_{H,x}(D))\,.
\end{equation}
For further background on the lexicographic product of graphs, see, e.g.,~\cite{MR2817074}.

Several papers in the literature studied the value of the domination number of a nontrivial lexicographic product of two graphs and determined the exact value in special cases, altogether giving the complete answer. Clearly, if $v$ is an isolated vertex in $G$, then $G[H]$ is isomorphic to the disjoint union of graphs $H$ and $(G-v)[H]$; in particular, this implies that $\gamma(G[H]) = \gamma(H) + \gamma((G-v)[H])$. It therefore suffices to consider the case when $G$ has no isolated vertices. \cite{MR2828067}~showed that if $\gamma(H) = 1$, then $\gamma(G[H]) = \gamma(G)$. In 2012, \v{S}umenjak et al.~showed that a nontrivial lexicographic product $G[H]$ of a connected graph $G$ and a connected graph $H$ with $\gamma(H)\ge 2$ satisfies $\gamma(G[H])=\gamma_t(G)$~\cite[Lemma 3.3]{vsumenjak2012roman}. By analyzing the proof of this result, it can be seen that the same equality holds whenever $G$ has no isolated vertices and $\gamma(H)\ge 2$. Therefore, the value of the domination number of the nontrivial lexicographic product of two graphs $G$ and $H$ such that $G$ is without isolated vertices is completely determined, as follows.

\begin{theorem}[combining results from~\cite{vsumenjak2012roman} and~\cite{MR2828067}]\label{thm:Gamma 2+}
If $G$ is a graph without isolated vertices and $H$ is any graph, then
$$\gamma(G[H])=\left\{
                 \begin{array}{ll}
                   \gamma(G), & \hbox{if $\gamma(H)= 1$;}\\
                   \gamma_t(G), & \hbox{if $\gamma(H)\geq 2$.}
                 \end{array}
               \right.$$
\end{theorem}

It is worth mentioning that these works were preceded by the observation that for every graph $G$ without isolated vertices, we have $\gamma(G[2K_1])= \gamma_t(G)$. This relation was noted by Kratsch and Stewart in 1997, who used it to reduce the total dominating set problem to the dominating set problem (\cite{MR1474143}). The proof of equality $\gamma(G[2K_1])= \gamma_t(G)$ is implicit in the proof of Lemma~2 from~\cite{MR1474143}.

Let us also remark that Theorem~6 from~\cite{Sitthiwirattham} regarding the value of the domination number of the lexicographic product of two graphs where the base graph is complete is not true. The theorem states that for every connected graph $H$, we have $\gamma(K_n[H])=\gamma(H)$. Theorem~\ref{thm:Gamma 2+} contradicts this. Namely, if $H$ is a graph with $\gamma(H)\geq 2$ and $n\ge 2$, then by Theorem~\ref{thm:Gamma 2+}, we have $\gamma(K_n[H])=\gamma_t(K_n)=2$. Hence, for every connected graph $H$ with $\gamma(H)\geq 3$ and every $n\ge 2$ the equality $\gamma(K_n[H])=\gamma(H)$ stated in Theorem 6 in~\cite{Sitthiwirattham} is false.

For the sake of completeness, we give in the following corollary a formula for the domination number of the lexicographic product of any two graphs.

\begin{corollary}\label{cor:isolated}	
Let $G$ and $H$ be any two graphs and let $I$ be the set of isolated vertices in $G$.
Then
$$\gamma(G[H])=\left\{
                 \begin{array}{ll}
                   |V(G)|\gamma(H), & \hbox{if $G$ is edgeless;}\\
                   \gamma(G), & \hbox{if $G$ has an edge and $\gamma(H)= 1$;}\\
                   \gamma_t(G-I)+|I|\gamma(H), & \hbox{if $G$ has an edge and $\gamma(H)\geq 2$.}
                 \end{array}
               \right.$$
\end{corollary}	

\begin{proof}	
Let $m = |I|$. If $G$ is edgeless, then $I = V(G)$ and $G[H]$ is isomorphic to the disjoint union of $m$ copies of $H$, hence
$\gamma(G[H]) = m\gamma(H)$.
Suppose now that $G$ has an edge. Setting $G' = G-I$, we observe
that the product $G[H]$ is isomorphic to the disjoint union of the product $G'[H]$
and $m$ copies of $H$, hence $\gamma(G[H]) = \gamma(G'[H]) + m\gamma(H)$.
By Theorem~\ref{thm:Gamma 2+}, we have
$$\gamma(G'[H])=\left\{
                 \begin{array}{ll}
                   \gamma(G'), & \hbox{if $\gamma(H)= 1$;}\\
                   \gamma_t(G'), & \hbox{if $\gamma(H)\geq 2$\,,}
                 \end{array}
               \right.$$
which implies
$$\gamma(G[H])=\left\{
                 \begin{array}{ll}
                   \gamma(G')+m, & \hbox{if $\gamma(H)= 1$;}\\
                   \gamma_t(G')+m\gamma(H), & \hbox{if $\gamma(H)\geq 2$\,,}
                 \end{array}
               \right.$$
Using the fact that $\gamma(G)=\gamma(G')+m$, the result follows.	
\end{proof}

\section{Reducible and irreducible dominating sets} \label{sec:irreducible}

An important notion for our characterization of minimal dominating sets of the lexicographic product graphs is the notion of irreducible dominating sets, which we now introduce.

\begin{definition}
A dominating set $D$ is said to be {\em reducible} if there exists a vertex $u\in D$ such that $D\setminus \{u\}$ is also a dominating set of $G$ and the sets of vertices that are totally dominated by $D$ and $D\setminus \{u\}$ coincide (that is, $N(D)=N(D\setminus \{u\})$). A dominating set is {\em irreducible} if it is not reducible.
\end{definition}

To illustrate these notions, we now list some easily verifiable facts and examples:
\begin{itemize}
\item A dominating set $D$ is reducible if and only if for some $u\in D$, the set $D\setminus \{u\}$ totally dominates $N[u]$.
\item Every minimal dominating set is irreducible. Consequently, every dominating set contains an irreducible dominating set.
\item Every minimal total dominating set is irreducible.
\item Suppose that $D$ is a dominating set inducing a subgraph of maximum degree at most one (that is, every vertex in $D$ has at most one neighbor in $D$). Then, $D$ is irreducible.
\item Consider a complete graph $K_n$ on $n$ vertices with $n\geq 3$. The only irreducible dominating sets in $K_n$ are those of size $1$ or $2$, that is, minimal dominating sets and minimal total dominating sets.
\end{itemize}

The next proposition compares a well-known characterization of minimal dominating sets with a characterization of irreducible dominating sets. In order to state the proposition, we need to introduce some more terminology. Given a set $D\subseteq V(G)$, a vertex $u\in D$ and a vertex $v\in V(G)$, we say that $v$ is a {\em $D$-private closed neighbor} of $u$ if $N[v]\cap D = \{u\}$.
Note that $u$ is a $D$-private closed neighbor of itself if and only if $u$ is barely dominated by $D$. If $u$ is totally dominated by $D$, then every $D$-private closed neighbor $v$ of $u$ is an element of $V(G)\setminus D$ that is not adjacent to any vertex in $D\setminus\{u\}$. Moreover, every vertex in $D$ with a unique neighbor in $D$ will be called a {\em $D$-leaf}.

\begin{proposition}\label{fact:irreducible}
For any graph $G$ and a dominating set $D\subseteq V(G)$, the following holds:
\begin{enumerate}
  \item $D$ is minimal if and only if every vertex in $D$ has a $D$-private closed neighbor.

  \item $D$ is irreducible if and only if every vertex in $D$ either has a $D$-private closed neighbor or is adjacent to a $D$-leaf.
\end{enumerate}
Consequently, $D$ is reducible if and only if it has a vertex that: (i) does not have any $D$-private closed neighbors and (ii) is not adjacent to any $D$-leaf.
\end{proposition}

\begin{proof}
Let $D$ be a dominating set in $G$. Suppose first that $D$ is minimal, let $u\in D$, and let $D' = D\setminus \{u\}$. If $u$ has no $D$-private closed neighbor, then for every vertex $v\in V(G)$, we have $N[v]\cap (D\setminus \{u\})\neq \emptyset$. Hence, $D'$ is a dominating set in $G$, contradicting the minimality of $D$. This shows that every vertex in $D$ has a $D$-private closed neighbor. Conversely, if every vertex $u\in D$ has a $D$-private closed neighbor, say $x_u$, then for every proper subset $D'\subset D$ and every $u\in D\setminus D'$,
we have $N[x_u]\cap D = \{u\}$, which implies $N[x_u]\cap D'\subseteq (N[x_u]\cap D)\setminus\{u\} = \emptyset$,
hence the set $D'$ does not dominate $x_u$. It follows that $D$ is a minimal dominating set.

Suppose next that $D$ is irreducible and suppose for the sake of contradiction that some vertex $u\in D$ does not have any $D$-private closed neighbor and is not adjacent to any $D$-leaf. Since $u$ does not have any $D$-private closed neighbor, the set $D\setminus \{u\}$ is also a dominating set of $G$. Thus, since $D$ is irreducible, the sets of vertices totally dominated by $D$ and $D\setminus \{u\}$ do not coincide, that is, there is a vertex $v\in N(D)\setminus (N(D\setminus \{u\}))$. Then $N(v)\cap D = \{u\}$ and since $u$ does not have any $D$-private closed neighbor, we infer that $v\in D$. It follows that $v$ is a $D$-leaf adjacent to $u$, contradicting the assumption that $u$ is not adjacent to any $D$-leaf.

Finally, suppose that every vertex in $D$ either has a $D$-private closed neighbor or is adjacent to a $D$-leaf and,
for the sake of contradiction, that $D$ is reducible. Then, there exists a vertex $u\in D$ such that
$D\setminus \{u\}$ is a dominating set of $G$ and $N(D)=N(D\setminus \{u\})$. The fact that
$D\setminus \{u\}$ is a dominating set implies that $u$ does not have a $D$-private closed neighbor,
and consequently $u$ is adjacent to some $D$-leaf, say $v$.
We thus have $N(v)\cap D = \{u\}$, which implies that $v\in N(D)\setminus (N(D\setminus \{u\}))$, contrary to the fact that
$N(D)=N(D\setminus \{u\})$. This completes the proof.
\end{proof}

\section{Minimal dominating sets in lexicographic product graphs}\label{sec:minimal}

In this section we investigate the structure of dominating sets and of minimal dominating sets in the lexicographic product of two graphs.

\subsection{Dominating sets}

The following lemma characterizes when a vertex in the lexicographic product graph is dominated by a given set.

\begin{lemma}\label{lem:dominating-vertices}
For graphs $G$ and $H$, a vertex $(g,h)\in V(G[H])$, and a set $D\subseteq V(G)\times V(H)$, the following conditions are equivalent:
\begin{enumerate}
  \item $D$ dominates $(g,h)$ in $G[H]$.
  \item Either $p_G(D)$ totally dominates $g$ in $G$ or $p_{H,g}(D)$ dominates $h$ in $H$.
\end{enumerate}
In particular, if $(g,h)$ is dominated by $D$ in $G[H]$, then $g$ is dominated by $p_G(D)$ in $G$.
\end{lemma}

\begin{proof}
The definition of the lexicographic product implies that the closed neighborhood of
a vertex $(g,h)$ in $G[H]$ is given by $N_{G[H]}[(g,h)] = (\{g\}\times N_H[h])\cup (N_G(g)\times V(H))$.
Hence, the condition that $D$ dominates $(g,h)$ in $G[H]$ is equivalent to the condition that
either $D\cap (\{g\}\times N_H[h])\neq \emptyset$ or $D\cap (N_G(g)\times V(H))\neq \emptyset$.
The former condition is equivalent to the condition that $p_{H,g}(D)\cap N_H[h]\neq \emptyset$
(that is, $p_{H,g}(D)$ dominates $h$ in $H$),
while the latter one is equivalent to the condition that $p_{G}(D)\cap N_G(g)\neq \emptyset$
(that is, $p_G(D)$ totally dominates $g$ in $G$).

It remains to show that if $(g,h)$ is dominated by $D$ in $G[H]$, then $g$ is dominated by $p_G(D)$ in $G$.
If $g$ is totally dominated by $p_G(D)$ in $G$, then it is also dominated.
If $g$ is not totally dominated by $p_G(D)$ in $G$, then by the above the set $p_{H,g}(D)$ dominates $h$ in $H$. In particular, $p_{H,g}(D)\neq \emptyset$, which implies that $g\in p_G(D)$. Hence, $g$ is dominated by $p_G(D)$ in $G$ also in this case.
\end{proof}

Next, we characterize the dominating sets in the lexicographic product of two graphs.
Recall that given a set $S\subseteq V(G)$, a vertex $x\in V(G)$ is said to be barely dominated by $S$ if
$N[x]\cap S = \{x\}$ (that is, if $x$ is an isolated vertex in the subgraph of $G$ induced by $S$).

\begin{sloppypar}
\begin{proposition}\label{lem:Dom}
For any two graphs $G$ and $H$, a set $D\subseteq V(G)\times V(H)$
is a dominating set in $G[H]$ if and only if
$p_G(D)$ is a dominating set in $G$ such that for every vertex $g\in V(G)$ that is barely dominated by $p_G(D)$,
the set $p_{H,g}(D)$ is a dominating set in $H$.
\end{proposition}
\end{sloppypar}

\begin{proof}
Suppose first that $D$ is a dominating set in $G[H]$. Fix an arbitrary vertex $h\in V(H)$. For every vertex $g\in V(G)$, set $D$ dominates vertex $(g,h)$ in $G[H]$. Hence, by Lemma~\ref{lem:dominating-vertices}, set $p_G(D)$ dominates $g$ in $G$. Since this holds for an arbitrary vertex of $G$,
we infer that $p_G(D)$ is a dominating set in $G$. Moreover, if $g\in V(G)$ is a vertex that is barely dominated by $p_G(D)$, then, since $(g,h)$ is dominated by $D$, Lemma~\ref{lem:dominating-vertices} implies that $p_{H,g}(D)$ dominates $h$ in $H$. Since this holds for
an arbitrary vertex $h$ of $H$, we infer that $p_{H,g}(D)$ is a dominating set in $H$, as claimed.

Conversely, suppose that $p_G(D)$ is a dominating set in $G$ such that for every vertex $g\in V(G)$ that is barely dominated by $p_G(D)$, the set $p_{H,g}(D)$ is a dominating set in $H$. Consider an arbitrary vertex $(g,h)\in V(G[H])$. If
$p_G(D)$ totally dominates $g$, then by Lemma~\ref{lem:dominating-vertices}, we conclude that $D$ dominates $(g,h)$. Suppose that $g\in V(G)$ is a vertex that is not totally dominated by $p_G(D)$. Since $p_G(D)$ is a dominating set in $G$, vertex $g$ is barely dominated by $D$. By assumption, $p_{H,g}(D)$ is a dominating set in $H$. Therefore, $p_{H,g}(D)$ dominates $h$ in $H$, and we again conclude that $D$ dominates $(g,h)$ using Lemma~\ref{lem:dominating-vertices}.
\end{proof}

Proposition~\ref{lem:Dom} leads to an alternative proof of Theorem~\ref{thm:Gamma 2+}.

\begin{proof}[of Theorem~\ref{thm:Gamma 2+}]
Let $G$ be a graph without isolated vertices and suppose first that $\gamma(H)=1$. Let $D$ be a minimum dominating set in $G[H]$. Then by Proposition~\ref{lem:Dom} it follows that $p_G(D)$ is a dominating set in $G$.
Hence $\gamma(G[H]) = |D|\ge |p_G(D)|\geq \gamma(G)$, where the first inequality follows from the definition of the projection map $p_G$ and the second one from the fact that $p_G(D)$ is a dominating set in $G$.
Let now $D_1$ be a minimum dominating set in $G$, let $h$ be a universal vertex in $H$, and let $D=D_1\times \{h\}$.
It follows from Proposition~\ref{lem:Dom} that $D$ is a dominating set in $G[H]$.
Hence $\gamma(G[H])\leq |D|=|D_1|=\gamma(G)$. Therefore $\gamma(G[H])=\gamma(G)$.
	
Suppose now that $\gamma(H)\ge 2$.  Let $D_t$ be a minimum total dominating set in $G$. Then by Proposition~\ref{lem:Dom}, it follows that $D_t\times \{h\}$ is a dominating set in $G[H]$, for every $h\in V(H)$. This shows that 	
\begin{equation}\label{eq:gamma_t}	
\gamma(G[H])\leq \gamma_t(G).	
\end{equation} 	
Let $D$ be a minimum dominating set of $G[H]$.  Let $A\subseteq p_G(D)$ be the set of vertices of $p_G(D)$ that are strongly dominated by $p_G(D)$, and let $B=p_G(D)\setminus A$, that is, $B$ is the set of vertices barely dominated by $p_G(D)$. It follows
from Proposition~\ref{lem:Dom} that $p_G(D)$ is a dominating set of $G$ and for every $x\in B$, the set $p_{H,x}(D)$  is a dominating set of $H$. We conclude that 	
\begin{equation*}	
|D|=\sum_{x\in p_{G}(D)} |p_{H,x}(D)|=\sum_{x\in A} |p_{H,x}(D)|+\sum_{x\in B} |p_{H,x}(D)|\geq |A|+\gamma(H)|B|
\geq |A|+2|B|\,.	
\end{equation*}	
Since $G$ has no isolated vertices, we can associate to every vertex $b\in B$ a neighbor $b'$ of $b$ in $G$.
Let $B'=\{b': b \in B\}$. Clearly, $|B'|\leq |B|$ and it is easy to see that
the set $A\cup B \cup B'$ is a total dominating set in $G$. We conclude that
\begin{equation}\label{eq:D}	
\gamma(G[H])=|D|\geq |A|+2|B|\geq |A|+|B|+|B'|\geq \gamma_t(G).	
\end{equation}
Inequalities~\eqref{eq:gamma_t} and \eqref{eq:D} imply that $\gamma(G[H])=\gamma_t(G)$.	
\end{proof}

\subsection{Minimal dominating sets}

We now develop a characterization of minimal dominating sets in the lexicographic product $G[H]$. In order to state it, we need some additional terminology. Let $D$ be an irreducible dominating set in a graph $G$. For a vertex $x\in D$ we say that $x$ is {\em $D$-redundant} if $D\setminus \{x\}$ is a dominating set in $G$ (equivalently, if $x$ does not have any $D$-private closed neighbors). Using this terminology, the characterization of irreducible dominating sets given by Proposition~\ref{fact:irreducible}
can be restated as follows: a dominating set $D$ is irreducible if and only if every $D$-redundant vertex is adjacent to a $D$-leaf.

The characterization of minimal dominating sets in the product graph $G[H]$ is given in Theorem~\ref{thm:MinDom}. The theorem gives necessary and sufficient conditions that a set $D\subseteq V(G[H])$ has to satisfy in order to be a minimal dominating set in $G[H]$. The conditions are expressed in terms of conditions on the set $p_G(D)$, which is a subset of $V(G)$, and the sets $p_{H,x}(D)$ for $x\in p_G(D)$, which are subsets of $V(H)$. Note that, due to relation~\eqref{eq:decomposition-D}, the theorem also yields a constructive way obtaining all minimal dominating sets in $D$.

\begin{theorem}\label{thm:MinDom}
For any two graphs $G$ and $H$, a set $D\subseteq V(G)\times V(H)$ is a minimal dominating set in $G[H]$ if and only if the following conditions hold:
\begin{enumerate}[(i)]
\item $p_G(D)$ is an irreducible dominating set in $G$;
\item for every vertex $x\in p_G(D)$,
$$\textrm{the set }p_{H,x}(D)\textrm{ is a }\left\{
                \begin{array}{ll}
                  \textrm{subset of size $1$ of $V(H)$}, & \hbox{if $x$ is totally dominated by $p_G(D)$ in $G$;} \\
                  \textrm{minimal dominating set in $H$}, & \hbox{if $x$ is barely dominated by $p_G(D)$ in $G$;}
                \end{array}
              \right.
$$
\item every $p_G(D)$-redundant vertex is adjacent to a $p_G(D)$-leaf $y$ such that the set $p_{H,y}(D)$ is not dominating in $H$.
\end{enumerate}
\end{theorem}

\begin{proof}
First, we establish necessity of the three conditions. Let $D\subseteq V(G)\times V(H)$ be a minimal dominating set in $G[H]$. We show each of the three conditions one by one.

$(i)$ Proposition~\ref{lem:Dom} implies that $p_G(D)$ is a dominating set in $G$. To prove that $(i)$ holds, it remains to prove that $p_G(D)$ is irreducible. Suppose that $p_G(D)$ is a reducible dominating set in $G$. Then there exists $u\in p_G(D)$ such that $p_G(D) \setminus \{u\}$ is a dominating set in $G$ and the sets of vertices of $G$ totally dominated by $p_G(D)$ and $p_G(D)\setminus \{u\}$ coincide. Let $D'=D\setminus (\{u\}\times V(H))$. Observe that $D'$ is a proper subset of $D$ and $p_G(D')=p_G(D)\setminus \{u\}$.
We claim that $D'$ is also a dominating set in $G[H]$. Suppose this is not the case, that is, there exists a vertex $(g,h)$ of $G[H]$ not dominated by $D'$.
Since $(g,h)$ is dominated by $D$ but not by $D'$, all elements of $N_{G[H]}[(g,h)]\cap D$ have $u$ as first coordinate.
If $g\neq u$, then we infer that $g$ is adjacent to $u$ in $G$; in particular,
$g$ is totally dominated by $p_G(D)$ in $G$. Since the sets of vertices of $G$ totally dominated by $p_G(D)$ and $p_G(D)\setminus \{u\}$ coincide, $g$ is totally dominated by $p_G(D')$, which further implies that $(g,h)$ is dominated by $D'$.
If $g=u$, then since $p_G(D')$ is a dominating set in $G$, it follows that $u$ is totally dominated by $p_G(D')$, and consequently, $(u,h)$ is dominated by $D'$. This shows that $D'$ is a dominating set and contradicts the assumption that $D$ is a minimal dominating set.
It follows that $p_G(D)$ is an irreducible dominating set in $G$.

$(ii)$ Let $x\in p_G(D)$ be totally dominated by $p_G(D)$. Since $x\in p_G(D)$, it follows that $p_{H,x}(D)$ is non-empty. We claim that it has size~$1$. Suppose to the contrary that $h_1,h_2\in p_{H,x}(D)$ for $h_1\neq h_2$. Let $D'=D\setminus \{(x,h_2)\}$. We claim that $D'$ is also a dominating set in $G[H]$. For $g\in V(G)\setminus \{x\}$ and $h\in V(H)$, it is clear that $(g,h)$ is dominated by $(x,h_1)$ if and only if it is dominated by $(x,h_2)$. Since $(x,h_1),(x,h_2)\in D$ and $D$ is a dominating set in $G[H]$, it follows that $(g,h)$ is dominated by $D'$. It remains to show that $D'$ dominates vertices with first coordinate equal to $x$. However, this easily follows from the fact that $x$ is totally dominated by $p_G(D)$. We conclude that $D'$ is a dominating set in $G[H]$. Since $D'$ is a proper subset of $D$, this contradicts the assumption that $D$ is a minimal dominating set in $G[H]$. The obtained contradiction shows that $|p_{H,x}(D)|=1$.

Suppose now that $x\in p_G(D)$ is barely dominated by $p_G(D)$. Then by Proposition~\ref{lem:Dom}, the set $p_{H,x}(D)$ is a dominating set in $H$.
Moreover, $p_{H,x}(D)$ is also a minimal dominating set in $H$, since otherwise, if $D_H$ is a dominating set in $H$ that is properly contained in $p_{H,x}(D)$,
one could apply Proposition~\ref{lem:Dom} to infer that the set $(D\setminus(\{x\}\times V(H)))\cup (\{x\}\times D_H)$ is a dominating set in $G[H]$ that is properly contained in $D$, contradicting the minimality of $D$. This establishes $(ii)$.

$(iii)$ Suppose that $(iii)$ fails, that is, there exists a $p_G(D)$-redundant vertex $x$ such that for every $p_G(D)$-leaf $y$ adjacent to $x$, the set $p_{H,y}(D)$ is dominating in $H$.
Let $D'=D\setminus (\{x\}\times V(H))$. We claim that $D'$ is also a dominating set in $G[H]$. Using Proposition~\ref{lem:Dom}, it is sufficient to verify that $p_G(D')$ is a dominating set in $G$ such that for every vertex $g\in V(G)$ that is barely dominated by $p_G(D')$,
the set $p_{H,g}(D')$ is a dominating set in $H$. Observe that $p_G(D')=p_G(D)\setminus \{x\}$.
The fact that $x$ is $D$-redundant implies that $p_G(D')$ is a dominating set in $G$.
Now, let $g\in V(G)$ be a vertex that is barely dominated by $p_G(D')$. Then $g\in p_{G}(D')$ and $g\in p_G(D)$. In particular, this implies that
$p_{H,g}(D')=p_{H,g}(D)$. If $g$ is barely dominated also by $p_G(D)$, then the fact that $D$ is a dominating set in $G[H]$ and Proposition~\ref{lem:Dom} imply that the set $p_{H,g}(D')$ is a dominating set in $H$.
If $g$ is not barely dominated by $p_G(D)$, then $g$ is a $p_G(D)$-leaf in $G$ adjacent to $x$. By assumption, the set $p_{H,g}(D')$ is a dominating set in $H$. By Proposition~\ref{lem:Dom}, we conclude that $D'$ is a dominating set, contradicting the assumption that $D$ is a minimal dominating set. The obtained contradiction shows that $(iii)$ holds.

\smallskip
In the rest of the proof, we show that the three conditions are also sufficient for $D$ to be a minimal dominating set. Suppose that conditions $(i)$--$(iii)$ hold. The fact that $D$ is a dominating set in $G[H]$ follows from Proposition~\ref{lem:Dom}.

It remains to prove minimality. Let $(g,h)\in D$ be arbitrary, and define the set $D'$ as $D'=D\setminus \{(g,h)\}$. If $g$ is barely dominated by $p_G(D)$, then condition $(ii)$ implies that $p_{H,g}(D)$ is a minimal dominating set in $H$.
Note that $h\in p_{H,g}(D)$. Since $p_{H,g}(D)\setminus \{h\}$ is not a dominating set in $H$, there exists a vertex $h'$ in $H$ that is not dominated by
$p_{H,g}(D)\setminus \{h\}$. It follows that vertex $(g,h')$ is not dominated by the set of vertices in $D'$ with first coordinate $g$.
Since $g$ is barely dominated by $p_G(D)$, no vertex in $G[H]$ with first coordinate $g$ is adjacent to a vertex in $D$ with first coordinate other than $g$. We infer that vertex $(g,h')$ is not dominated by $D'$, contradicting the assumption that $D'$ is a dominating set in $G[H]$.

Next suppose that $g$ is totally dominated by $p_G(D)$. Since $(g,h)\in D$ and $g$ is totally dominated by $D$, condition $(ii)$ implies that there exists a vertex $h\in V(H)$ such that $p_{H,g}(D)=\{h\}$. We conclude that $p_G(D')=p_G(D)\setminus \{g\}$. If $g$ is not $p_G(D)$-redundant, then $p_G(D)\setminus \{g\}$ is not a dominating set in $G$, and by Proposition~\ref{lem:Dom} it follows that $D'$ is not a dominating set in $G[H]$. Now suppose that $g$ is $p_G(D)$-redundant. By condition $(iii)$, vertex $g$ is adjacent to a $p_G(D)$-leaf, say $g_1$, such that $p_{H,g_1}(D)$ is not a dominating set in $H$. Since $p_{H,g_1}(D') = p_{H,g_1}(D)$ and $g_1$ is barely dominated by $D'$, it follows from Proposition~\ref{lem:Dom} that $D'$ is not a dominating set in $G[H]$. We conclude that no proper subset of $D$ is a dominating set in $G[H]$ and hence $D$ is a minimal dominating set in $G[H]$. This completes the proof.
\end{proof}

Note that assuming condition $(ii)$ in Theorem~\ref{thm:MinDom}, condition $(iii)$ can be equivalently stated as follows: every $p_G(D)$-redundant vertex is adjacent to a $p_G(D)$-leaf $y$ such that the unique vertex in $p_{H,y}(D)$ is not universal in $H$. This condition is trivially satisfied if $\gamma(H)\ge 2$. Therefore, for the case when $\gamma(H)\ge 2$, Theorem~\ref{thm:MinDom} takes on the following simpler formulation.

\begin{theorem}\label{thm:MinDom_gammaH=2}
For any two graphs $G$ and $H$ with $\gamma(H)\ge 2$, a set $D\subseteq V(G)\times V(H)$ is a minimal dominating set in $G[H]$ if and only if the following conditions hold:
\begin{enumerate}[(i)]
\item $p_G(D)$ is an irreducible dominating set in $G$;
\item for every vertex $x\in p_G(D)$,
$$\textrm{the set }p_{H,x}(D)\textrm{ is a }\left\{
                \begin{array}{ll}
                  \textrm{subset of size $1$ of $V(H)$}, & \hbox{if $x$ is totally dominated by $p_G(D)$ in $G$;} \\
                  \textrm{minimal dominating set in $H$}, & \hbox{if $x$ is barely dominated by $p_G(D)$ in $G$.}
                \end{array}
              \right.
$$
\end{enumerate}
\end{theorem}

It might be worth pointing out that in general, conditions $(i)$ and $(ii)$ alone do not imply condition $(iii)$. This is shown by the following example.

\begin{example}
Let $G$ be the $5$-vertex path with vertices $g_1\,g_2\,g_3\,g_4\,g_5$ along the path and let $H$ be the $3$-vertex path with vertices
$h_1\,h_2\,h_3$ along the path. Let  $D=\{(g_2,h_2),(g_3,h_1),(g_4,h_2)\}$. Then $D$ is a dominating set in $G[H]$. Moreover, we have $p_G(D)=\{g_2,g_3,g_4\}$ and the set $p_G(D)$ is an irreducible dominating set in $G$. It follows that condition $(i)$ from Theorem~\ref{thm:MinDom} holds. Observe that every vertex from $p_G(D)$ is totally dominated by  $p_G(D)$. It is now easy to see that condition $(ii)$ from Theorem~\ref{thm:MinDom} holds. However, it is not difficult to verify that $D'=\{(g_2,h_2),(g_4,h_2)\}$, which is a proper subset of $D$, is also a dominating set in $G[H]$. We conclude that $D$ is not a minimal dominating set in $G[H]$, but satisfies conditions $(i)$ and $(ii)$ from Theorem~\ref{thm:MinDom}.
\end{example}

For later use, we establish in the following proposition a lower bound on the upper domination number $\Gamma$ of the lexicographic product of two graphs.

\begin{proposition}\label{prop:upperdomination}
For every two graphs $G$ and $H$, it holds that $\Gamma(G[H])\geq \alpha(G)\Gamma(H)$.
\end{proposition}

\begin{proof}
Let $S$ be a maximum independent set in $G$, let $A$ be a minimal dominating set in $H$ of size $\Gamma(H)$, and let $D=S\times A$.
Since $|S| = \alpha(G)$ and $|A| = \Gamma(H)$, Theorem~\ref{thm:MinDom} implies that $D$ is a minimal dominating set in $G[H]$ of size $\alpha(G)\Gamma(H)$. This shows that $\Gamma(G[H])\geq \alpha(G)\Gamma(H)$.
\end{proof}

\begin{remark}
The value of $\Gamma(G[H])$ cannot be bounded from above by any function of the product $\alpha(G)\Gamma(H)$.
For example, let $G$ be the graph obtained from two copies of $K_n$ (for $n\geq 4$) joined by a perfect matching and let $H=C_4$.
Then $\Gamma(G[H])= n$ and $\alpha(G)\Gamma(H)=2\cdot 2=4$.
\end{remark}

\section{Well-dominated lexicographic product graphs}

In this section we characterize well-dominated nontrivial lexicographic product graphs.
Recall that a graph is well-dominated if all of its minimal dominating sets are of the same size. A lexicographic product graph $G[H]$ is connected if and only if $G$ is connected
(see, e.g., Corollary 5.14~in~\cite{MR2817074}). In particular, if $G$ has components $G_1,\ldots, G_k$, then the components of $G[H]$ are $G_1[H],\ldots, G_k[H]$. It is not difficult to see that a graph is well-dominated if and only if all of its components are well-dominated. Therefore, when characterizing nontrivial lexicographic product graphs that are well-dominated, we may without loss of generality restrict our attention to the case of nontrivial products $G[H]$ such that $G$ is connected. The following theorem states the corresponding characterization.

\begin{theorem}\label{thm:Well-dominated Lex}
A nontrivial lexicographic product, $G[H]$, of a connected graph $G$ and a graph $H$
is well-dominated if and only if one of the following conditions holds:
\begin{enumerate}[(i)]
\item $G$ is well-dominated and $H$ is complete, or
\item $G$ is complete and $H$ is well-dominated with $\gamma(H)=2$.
\end{enumerate}
\end{theorem}

\begin{proof}
We start by establishing the simpler direction, namely that each of the two conditions is sufficient for the product graph to be well-dominated. Suppose first that $G$ is a connected well-dominated graph and $H$ is complete.
Let $D$ be a minimal dominating set in $G[H]$. Since $H$ is complete, conditions $(i)$ and $(ii)$ of Theorem~\ref{thm:MinDom} imply that $p_G(D)$ is an irreducible dominating set in $G$ such that for each vertex $x\in p_G(D)$, we have $|p_{H,x}(D)| = 1$.
Consequently, $|D|=|p_G(D)|$. We claim that $p_G(D)$ is in fact a minimal dominating set in $G$.
If this were not the case, then $p_G(D)$ would contain a $p_G(D)$-redundant vertex, say $x$. By condition $(iii)$ of
Theorem~\ref{thm:MinDom}, vertex $x$ is adjacent to a $p_G(D)$-leaf $y$ such that the set
$p_{H,y}(D)$ is not dominating in $H$. However, this contradicts the fact that $p_{H,y}(D)$ is a non-empty set in a complete graph.
It follows that no vertex in $p_G(D)$ is $p_G(D)$-redundant; hence, $p_G(D)$ is a minimal dominating set.
Let $D'$ be another minimal dominating set in $G[H]$. Then, by the same arguments, $p_G(D')$ is also a minimal dominating set in $G$, and $|D'|=|p_G(D')|$. Since $G$ is well-dominated, it follows that $|p_G(D)|=|p_G(D')|$ and consequently $|D|=|D'|$.
This shows that any two minimal dominating sets in $G[H]$ are of the same size; hence $G[H]$ is a well-dominated graph.

Second, suppose that $G$ is complete and $H$ is well-dominated with $\gamma(H)=2$. We claim that all minimal dominating sets in $G[H]$ are of size $2$. Let $D$ be a minimal dominating set in $G[H]$. Then by Theorem~\ref{thm:MinDom_gammaH=2} the set $p_{G}(D)$ is an irreducible dominating set in $G$. Since $G$ is complete, any irreducible dominating set in $G$ is of size one or two. Suppose first that $|p_{G}(D)|=1$ and let $p_{G}(D)=\{x\}$. By Theorem~\ref{thm:MinDom_gammaH=2}, the set $p_{H,x}(D)$ is a minimal dominating set in $H$. Since $H$ is well-dominated with $\gamma(H)=2$, it follows that $|p_{H,x}(D)|=2$, and therefore $|D|=2$. Suppose now that $|p_{G}(D)|=2$ and let $p_{G}(D)=\{x,y\}$. Then $x$ and $y$ are both totally dominated by $p_{G}(D)$ and therefore, again by Theorem~\ref{thm:MinDom_gammaH=2}, $|p_{H,x}(D)|=|p_{H,y}(D)|=1$. We conclude that $|D|=2$. This shows that all minimal dominating sets in $G[H]$ are of size $2$, as claimed. Therefore, $G[H]$ is well-dominated with $\gamma(H)=2$.

\smallskip
It remains to show that the disjunction of the two conditions is also necessary for the product graph to be well-dominated. Suppose that the product $G[H]$ is well-dominated. First, we show that $H$ is well-dominated. Suppose this is not the case and let $D_1$ and $D_2$ be two minimal dominating sets of $H$ with $|D_1|\neq |D_2|$.
Let $S$ be a maximal independent set in $G$.
We claim that $S\times D_1$ and $S\times D_2$ are minimal dominating sets of $G[H]$.  Observe that $p_G(S\times D_i)=S$ for $i\in \{1,2\}$, and since $S$ is a maximal independent set, it is also a minimal dominating set in $G$, and hence $S$ is an irreducible dominating set in $G$. Moreover, since $S$ is minimal dominating set, it follows that there are no $S$-redundant vertices.
Observe that for every $x\in S$ we have $p_{H,x}(S\times D_i)=D_i$. Hence,  Theorem~\ref{thm:MinDom} implies that
$S\times D_1$ and $S\times D_2$ are minimal dominating sets in $G[H]$. The fact that
$S\times D_1$ and $S\times D_2$ are of different cardinalities now contradicts the assumption that $G[H]$ is well-dominated.

We consider three cases depending on the value of $\gamma(H)$. Suppose first that $\gamma(H)=1$. Since $H$ is well-dominated, all minimal dominating sets in $H$ are of size $1$; hence, $H$ is a complete graph. We claim that $G$ is well-dominated (and thus condition $(i)$ will hold). Let $D_1$ and $D_2$ be two minimal dominating sets in $G$. The minimality of $D_1$ and $D_2$ implies that $G$ contains no $D_1$-redundant (resp.~$D_2$-redundant) vertices. Let $h$ be an arbitrary vertex of $H$. By Theorem~\ref{thm:MinDom}, the sets $D_1\times \{h\}$ and $D_2\times \{h\}$ are minimal dominating sets in $G[H]$. Since all minimal dominating sets in $G[H]$ are of the same size, it follows that $|D_1|=|D_2|$; hence, $G$ is well-dominated, as claimed.

Suppose now that $\gamma(H) = 2$. We claim that in this case $G$ is complete (and thus condition $(ii)$ will hold).
Suppose, to the contrary, that $G$ is not complete. Then, since $G$ is connected, it contains a pair of vertices, say $x$ and $y$, at distance two.
Let $u$ be a common neighbor of $x$ and $y$. Let $S$ be a maximal independent set in $G$ containing $x$ and $y$ and let $A$ be a minimum dominating set in $H$.
Since $S$ is a maximal independent set in $G$, it follows that $S$ is a minimal dominating set in $G$. Moreover, as there are no edges between vertices inside $S$, it follows that every vertex of $S$ is barely dominated by $S$.
Thus, we can apply Theorem~\ref{thm:MinDom_gammaH=2} to infer that $S\times A$ is a minimal dominating set in $G[H]$; its size is $2|S|$.
Clearly, the set $S\cup \{u\}$ is a dominating set in $G$.
Now, let $S'$ be an inclusion-minimal subset of $S\cup \{u\}$ such that $N[S'] = N[S\cup \{u\}]$ (that is, $S'$ is dominating in $G$) and $N(S') = N(S\cup \{u\})$. Since $u$ is the only neighbor of $x$ in $S\cup \{u\}$, we infer that $u\in S'$.
Moreover, since $u\in N(S\cup \{u\})$, we have $u\in N(S')$. It follows that $S'\cap N[u]$ contains $u$ and at least one neighbor of $u$.
Take any vertex $h\in V(H)$ and consider the set $D$ defined with
$D = ((S'\cap N[u]) \times \{h\})\cup((S'\setminus N[u])\times A)$.
Since every vertex in $S'\cap N[u]$ is totally dominated by $S'$ and every vertex in $S'\setminus N[u]$ is barely dominated by $S'$, Theorem~\ref{thm:MinDom_gammaH=2} implies that $D$ is a minimal dominating set in $G[H]$.
Since $G[H]$ is well-dominated, the cardinality of $D$ must be equal to that of $S\times A$.
Therefore, $|N[u]\cap S'|+2|S'\setminus N[u]| = 2|S|$, or, equivalently, $2|S'|= 2|S|+|N[u]\cap S'|$.
Since $|N[u]\cap S'|\ge 2$ and $|S'|\le |S|+1$, we obtain
$2|S'|\le 2(|S|+1)\le 2|S|+|N[u]\cap S'|= 2|S'|$, therefore equalities must hold throughout.
It follows that $|N[u]\cap S'|= 2$ and $|S'|= |S|+1$, which implies that $S' = S\cup \{u\}$ and therefore
$3=|\{u,x,y\}|\le |N[u]\cap S'|$, a contradiction.

Finally, suppose that $\gamma(H)\geq 3$. Since $\gamma_t(G)\leq 2\gamma(G)$ (see, e.g.,~\cite{bollobas1979graph}) and $\gamma(G) \le \alpha(G)$, we have $\gamma_t(G)\leq 2\alpha(G)$. In addition, Theorem~\ref{thm:Gamma 2+} gives $\gamma(G[H])=\gamma_t(G)$; therefore, $\gamma(G[H]) \le 2\alpha(G)$. Moreover, by Proposition~\ref{prop:upperdomination} it follows that $\Gamma(G[H])\geq \alpha(G)\Gamma(H)\geq 3\alpha(G)$. Since $G[H]$ is well-dominated, it follows that $\Gamma(G[H])=\gamma(G[H])$.
We obtain that
$$
\gamma_t(G)=\gamma(G[H])= \Gamma(G[H])\geq 3\alpha(G),
$$
a contradiction with $\gamma_t(G)\leq 2\alpha(G)$. We conclude that $G[H]$ is not well-dominated whenever $\gamma(H)\geq 3$.
This completes the proof.
\end{proof}

\section{Well-dominated graphs with small domination number}\label{sec:small-gamma}

Theorem \ref{thm:Well-dominated Lex} motivates the following question: What are the well-dominated graphs with domination number two? We address this question by giving a characterization of such graphs. We first recall some basic terminology. A {\em clique} in a graph is a set of pairwise adjacent vertices. A clique is {\em maximal} if it is not contained in any larger clique. A {\em triangle} in a graph $G$ is a clique of size three. A graph is said to be {\em triangle-free} if it has no triangles.

Since every maximal independent set in a graph is a minimal dominating set, a well-dominated graph with $\gamma(G)=2$ is also well-covered with $\alpha(G) = 2$. The following simple lemma characterizes well-covered graphs with $\alpha(G) = 2$.

\begin{lemma}\label{lem:well-covered-2}
A graph $G$ is well-covered with $\alpha(G) = 2$ if and only if its complement,
$\overline{G}$, is a triangle-free graph without isolated vertices.
\end{lemma}

\begin{proof}
A graph $G$ is well-covered with $\alpha(G) = 2$ if and only if all maximal cliques of $\overline{G}$ are of size two. This condition is equivalent to
$\overline{G}$ being triangle-free and without isolated vertices.
\end{proof}

To state the characterization of well-dominated graphs with $\gamma = 2$, we need to introduce some more notation. For a subset $X$ of the vertex set of a graph $G$, we denote
by $\overline{N[X]}$ the set $V(G)\setminus N[X]$. For two graphs $G$ and $H$, we say that a set $S\subseteq V(G)$ {\em induces an $H$} if the subgraph of $G$ induced by $S$ is isomorphic to $H$.

\begin{theorem}\label{thm:characterization gamma=2}
A graph $G$ is well-dominated with $\gamma(G) = 2$ if and only if the following conditions hold:
\begin{enumerate}[(i)]
\item $\overline{G}$ is a triangle-free graph without isolated vertices.
\item For every two triangles $T$ and $T'$ in $G$ such that $T\cup T'$ induces a $\overline{C_6}$, the set $T \cup \overline{N[T']}$ is not a dominating set in $G$.
\end{enumerate}
\end{theorem}

\begin{proof}
First, we establish necessity of the two conditions. Let $G$ be a well-dominated graph with $\gamma(G) = 2$. Condition $(i)$ follows from Lemma~\ref{lem:well-covered-2}.
Now, consider a pair of triangles $T$ and $T'$ in $G$ such that $T\cup T'$ induces a $\overline{C_6}$.
Suppose for a contradiction that $D =T \cup \overline{N[T']}$ is a dominating set in $G$. Then there exists a minimal dominating set $D'$ such that $D' \subseteq D$. Since the vertices of $T'$ cannot be dominated by $\overline{N[T']}$, they have to be dominated by the vertices of $T$. However, since $T\cup T'$ induces a $\overline{C_6}$, no vertex of $T$ dominates two vertices of $T'$, and therefore all the three vertices of $T$ must be in $D'$. Therefore, $|D'|\ge 3$, contradicting the assumption that $G$ is well-dominated with $\gamma(G)=2$.

\smallskip
Now, we establish sufficiency. Suppose that $G$ is a graph satisfying conditions $(i)$ and $(ii)$.
By Lemma~\ref{lem:well-covered-2}, condition $(i)$ implies that $G$ is well-covered with $\alpha(G) = 2$.
Since $\overline{G}$ does not have any isolated vertices, $G$ does not have any universal vertices, thus
$\gamma(G)\ge 2$. This inequality, combined with the inequality $\gamma(G)\le \alpha(G)$ and $\alpha(G) = 2$, implies $\gamma(G) = 2$.
Suppose for a contradiction that $G$ is not well-dominated. Then, $G$ contains a minimal dominating set $D$ of size at least $3$.
Let $a,b,c \in D$ be three distinct vertices in $D$. Since $D$ is minimal, each one of $a$, $b$, and $c$ has a $D$-private closed neighbor, say $a'$, $b'$, and $c'$, respectively, where the three vertices $a'$, $b'$, and $c'$ are pairwise distinct.
Let $T = \{a,b,c\}$ and $T' = \{a',b',c'\}$. Note that for every $t\in T$, if its $D$-private closed neighbor $t'$ is in $T$, then $t' = t$.
Next, observe that $T$ is a triangle in $G$ since if vertices $a$ and $b$ were non-adjacent (say), then
$\{a,b,c'\}$ would be an independent set of size $3$ in $G$, contradicting $\alpha(G) = 2$.
The fact that $T$ is a triangle and the definition of $T'$ imply that $D\cap T' = \emptyset$; in particular, $T\cap T' = \emptyset$.
A similar argument as the one applied earlier to $T$ shows that $T'$ is a triangle. Therefore, $T\cup T'$ induces a $\overline{C_6}$. It now suffices to show that $D\subseteq T \cup \overline{N[T']}$, as this will imply that $T \cup \overline{N[T']}$ is a dominating set in $G$ and
contradict condition $(ii)$.
Let $v\in D$. Clearly, if $v\in T$, then $v\in T \cup \overline{N[T']}$.
So let $v\in D\setminus T$. Suppose that
$v\not\in (T \cup \overline{N[T']})$.
Then $v\in N(T')\setminus T$, that is, there exists a vertex in $T'$, say (w.l.o.g.) $a'$, such that $a'$ is adjacent to $v$. However, this is impossible since $a'$ is a $D$-private closed neighbor of $a$, cannot be adjacent to $v\in D\setminus\{a\}$.
This completes the proof.
\end{proof}

To the best of our knowledge, the computational complexity of recognizing well-dominated graphs is open. Since conditions $(i)$ and $(ii)$ in Theorem~\ref{thm:characterization gamma=2} are polynomially testable, Theorem~\ref{thm:characterization gamma=2} implies the following partial result.

\begin{corollary}
The problem of recognizing well-dominated graphs can be solved in polynomial time in the class of graphs
$G$ with $\gamma(G) = 2$.
\end{corollary}

More generally, we now argue that for every fixed $k$, the problem of recognizing well-dominated graphs can be solved in polynomial time in the class of
graphs $\{G: \gamma(G)=k\}$. We first recall some terminology related to hypergraphs (see, e.g.,~\cite{MR1013569}). A \emph{hypergraph} $\mathcal{H}$ is a pair $(V,\cal{E})$ where $V=V({\cal H})$ is a finite set of vertices and $\mathcal{E}=E({\cal H})$ is a set of subsets of $V$, called \emph{hyperedges}. A vertex set $X\subseteq V$ is called a \emph{transversal} of $\mathcal{H}$ if $X$ intersects every hyperedge of $\mathcal{H}$, and it is called a \emph{minimal transversal} if it is a transversal that does not properly contain any other transversal.
Let $\mathcal{H}^*$ denote the hypergraph with vertex set $V({\cal H})$ having as hyperedges exactly the minimal transversals of $\mathcal{H}$.
A hypergraph is said to be \emph{Sperner} (or: a \emph{clutter}) if no hyperedge of $\mathcal{H}$ contains another hyperedge.

The {\sc Hypergraph Transversal} problem is the decision problem that takes as input two Sperner hypergraphs ${\cal H}$ and ${\cal H}'$
and asks whether ${\cal H}' = {\cal H}^*$. This is a well studied problem whose computational complexity status is a notorious open problem. As shown by~\cite{MR1417667}, the problem admits a quasi-polynomial-time solution (an algorithm running in time $n^{o(\log n)}$ where $n$ is the total input size). Moreover, several special cases have been shown to be solvable in polynomial time. For our purpose, polynomial-time solvability of the following special case will be useful, shown by~\cite{MR1361157} and by~\cite{MR1754735} (in the equivalent context of dualization of monotone Boolean functions):

\begin{theorem}[\cite{MR1361157},~\cite{MR1754735}]\label{thm:hyp-trans-const-k}
For every positive integer $k$,  the {\sc Hypergraph Transversal} problem is solvable in polynomial time if all hyperedges of one of the two hypergraphs ${\mathcal H}$ and ${\mathcal H}'$ are of size at most $k$.
\end{theorem}

Theorem~\ref{thm:hyp-trans-const-k} has the following consequence:

\begin{corollary}\label{cor:hyp-trans-const-k}
For every positive integer $k$, the following problem is solvable in polynomial time:
Given a Sperner hypergraph $\mathcal{H}$, determine whether all minimal transversals of $\mathcal{H}$ are of size $k$.
\end{corollary}

\begin{proof}
We proceed as follows: first, we generate all $O(|V({\mathcal{H}})|^k)$ subsets of size $k$ of $V({\mathcal{H}})$ and test for each of them whether it is a minimal transversal of ${\mathcal{H}}$; this way, we obtain a hypergraph ${\mathcal H}'$.
The problem now becomes that of testing whether ${\mathcal H}' = {\mathcal H}^*$.
Since all hyperedges of ${\mathcal H}'$ are of size $k$, Theorem~\ref{thm:hyp-trans-const-k} implies that
the problem is indeed polynomially solvable.
\end{proof}

The announced result about the recognition of well-dominated graphs with small domination number can now be derived from Corollary~\ref{cor:hyp-trans-const-k}.

\begin{theorem}\label{thm:const-k}
For every positive integer $k$, the problem of recognizing well-dominated graphs can be solved in polynomial time in the class of graphs
$G$ with $\gamma(G) = k$.
\end{theorem}

\begin{proof}
Let $G=(V,E)$ be a graph with $\gamma(G) = k$. Consider the hypergraph $\mathcal{H}_G=(V,\mathcal{E})$, where $\mathcal{E}$ contains the inclusion-minimal elements of $\{N[v]: v\in V\}$. Observe that $\mathcal{H}_G$ is Sperner and that the minimal transversals of ${\cal H}_G$ are exactly the
minimal dominating sets of $G$. It follows that $G$ is well-dominated if and only if all minimal transversals of ${\cal H}_G$
are of size $k$. By Corollary~\ref{cor:hyp-trans-const-k}, this condition can be tested in polynomial time.
\end{proof}

\section{Concluding remarks}\label{sec:conclusion}

We introduced in this paper the notion of an irreducible dominating set, a variant of dominating set generalizing both minimal dominating  and minimal total dominating sets. The main application of this notion was a characterization of the minimal dominating sets in nontrivial lexicographic product graphs, which led to a complete characterization of nontrivial lexicographic product graphs that are well-dominated.

We believe that the notions studied in this paper deserve to be investigated further. In particular, we feel it would be interesting to develop a better understanding of the structure of irreducible dominating sets in general graphs, which might lead to further applications of this notion.
For example, since every minimal dominating set as well as every minimal total dominating set in a graph $G$ is
an irreducible dominating set, the following problem naturally arises:

\begin{problem}
Characterize the graphs $G$ such that every irreducible dominating set in $G$ is either
a minimal dominating set or a minimal total dominating set.
\end{problem}

Another related question is that of determining an expression for the upper domination number (the maximum size of a minimal dominating set) of a lexicographic product graph in terms of parameters of its factors. Furthermore, does a similar approach as the one used in this paper lead to characterizations of minimal ``dominating'' sets in lexicographic product graphs with respect to other types of domination? For example, a characterization of minimal total dominating sets in the nontrivial lexicographic product graphs might lead to a characterization of lexicographic product graphs that are well-totally-dominated, where a graph without isolated vertices
is said to be {\it well-totally-dominated} if all its minimal total dominating sets are of the same size (\cite{MR1605080}).

\begin{sloppypar}
\begin{problem}
Characterize the nontrivial lexicographic product graphs that are well-totally-dominated.
\end{problem}
\end{sloppypar}

Finally, let us remark that, to the best of our knowledge, the computational complexity of recognizing well-dominated graphs is in general still open.

\acknowledgements
\label{sec:ack}

The authors are grateful to an anonymous referee for helpful remarks.

\end{document}